\documentclass{article}
\usepackage{fullpage}
\usepackage{graphicx}
\usepackage{amsthm}
\usepackage{xspace}
\usepackage{algorithm}
\usepackage{algorithmicx}
\usepackage{framed}
\usepackage{xifthen}
\usepackage{tikz}
\usetikzlibrary{calc}
\usepackage{hyperref}
\hypersetup{
  colorlinks   = true, 
  urlcolor     = blue, 
  linkcolor    = blue, 
  citecolor   = red 
}

\usepackage{algpseudocode}

\usepackage{caption}
\usepackage{amssymb}
\usepackage{fleqn}
\usepackage{amsmath}
\newcommand{\defproblem}[3]{
  \vspace{1mm}
  \noindent\fbox{
    \begin{minipage}{0.96\textwidth}
      \begin{tabular*}{\textwidth}{@{\extracolsep{\fill}}lr} #1 \\ \end{tabular*}
      {\bf{Input:}} #2  \\
      {\bf{Task:}} #3
    \end{minipage}
  }
  \vspace{1mm}
}
\theoremstyle{plain} 
\newtheorem{theorem}{Theorem}[section] %
\newtheorem{definition}{Definition}[section]
\newtheorem{corollary}{Corollary}[section]
\newtheorem{lemma}{Lemma}[section]
\newtheorem{proposition}{Proposition}[section]

\newcommand{\SM}{\textsc{Stable Matching}\xspace}
\newcommand{\SR}{\textsc{Stable Roomates}\xspace}

\newcommand{\DSM}{\textsc{Disjoint Stable Matchings}\xspace}
\newcommand{\Men}{\ensuremath{\mathcal{M}}\xspace}
\newcommand{\Women}{\ensuremath{\mathcal{W}}\xspace}
\newcommand{\EGS}{\textsc{GS-Extended}\xspace}
\newcommand{\threecol}{\textsc{3-COLORING}\xspace}

\usepackage{lipsum}
\usepackage{amsfonts}
\usepackage{graphicx}
\usepackage{epstopdf}
\ifpdf
  \DeclareGraphicsExtensions{.eps,.pdf,.png,.jpg}
\else
  \DeclareGraphicsExtensions{.eps}
\fi

\title{Disjoint Stable Matchings\thanks{A preliminary version appeared in the proceedings of the 47th International Workshop on Graph-Theoretic Concepts in Computer Science (WG 2021). A part of the work was done when the first two authors were students at Chennai Mathematical Institute.
}}

\usepackage{authblk}

\title{Disjoint Stable Matchings\thanks{A preliminary version appeared in the proceedings of 47th Workshop in Graph-Theoretic Concepts in Computer Science}}

\author[1]{Aadityan Ganesh\thanks{Part of the work was done when the author was a student at Chennai Mathematical Institute}}
\author[2]{Vishwa Prakash HV\thanks{Part of the work was done when the author was a student at Chennai Mathematical Institute}}
\author[3]{Prajakta Nimbhorkar}
\author[3]{Geeverghese Philip}

\affil[1]{Princeton University, USA}
\affil[2]{Keio University, Japan}
\affil[3]{Chennai Mathematical Institute, India}
\date{\today}

\begin{document}
%
\maketitle              
\begin{abstract}
  We show that given a \SM instance \(G\) 
  as input, we can find a \emph{largest collection} of pairwise edge-disjoint
  \emph{stable matchings} of \(G\) in time linear in the input
  size. 
  This extends two classical results:
  \begin{enumerate}
  \item The Gale-Shapley algorithm, which can find at most two (``extreme'')
    pairwise edge-disjoint stable matchings of \(G\) in linear time, and
  \item The polynomial-time algorithm for finding a largest collection of
    pairwise edge-disjoint \emph{perfect matchings} (without the stability
    requirement) in a bipartite graph, obtained by combining K\"{o}nig's
    characterization with Tutte's \(f\)-factor algorithm.
  \end{enumerate}

  Moreover, we also give an algorithm to enumerate all maximum-length chains of
  disjoint stable matchings in the lattice of stable matchings of a given
  instance. This algorithm takes time polynomial in the input size for
  enumerating each chain. We also derive the expected number of such chains in a
  random instance of \SM.

For the \SR\ problem, we show that determining if a given instance admits at most two
disjoint stable matchings or three disjoint stable matchings is NP-complete.
\end{abstract}

\section{Introduction}\label{sec:introduction}
All our graphs are finite, undirected, and simple. We use \(V(G), E(G)\) to
denote the vertex and edge sets of a graph \(G\), respectively. A \emph{matching}
in a graph \(G\) is any subset \(M \subseteq E(G)\) of edges of \(G\) such that
no two edges in \(M\) have a common end-vertex. \textcolor{black}{We consider the \SM\ and \SR\ 
problems in this paper.} An input instance of the \SM
problem contains a bipartite graph \(G\) with the vertex
partition \(V(G) = \Men \uplus \Women\) where the two sides \(\Men, \Women\) are
customarily called ``the set of men'' and ``the set of women'', respectively.
Each woman has a strictly ordered preference list containing her neighbors---a
woman prefers to be matched with a man who comes earlier in her list, than with
one who comes later---and each man similarly has a strictly ordered preference
list containing all his neighbors. \textcolor{black}{An instance of the \SR\ problem consists
of a graph $H$ on a set of vertices referred to as a set of {\em agents} $\mathcal{A}$. The edge
set of $H$ denotes mutually acceptable pairs of agents. The following definition applies to both
the \SM\ and \SR\ problems.}
\begin{definition}[Blocking pair]
A pair of vertices  \(a,b\) in $G$ or $H$ is said to be a 
\emph{blocking pair} with respect to a matching \(M\) if both $a$ and $b$ 
prefer each other over their respective matched partners in $M$. 
\end{definition}

\begin{definition}[Stable matching]
A matching \(M\) is said to be \emph{stable} if there is no
blocking pair with respect to $M$.
\end{definition}
A matching $M$ that is not stable is said to be {\em unstable}. The \SM instance
consists of a bipartite graph \(G\) with vertex partition \(\Men \uplus \Women\)
and the associated preference lists. The \SM\ problem involves deciding if \(G\)
has a stable matching, and outputting one if it exists. \textcolor{black}{Similarly, the \SR\ problem instance consists
of the graph $H$ and the associated preference lists.}

The \SM\ \textcolor{black}{and \SR\ problems model a number of real-world applications. Some examples of the \SM\ problem 
consist of situations} where there are two disjoint
sets of entities---fresh graduates and intern positions; students and hostel
rooms; internet users and CDN servers; and so on---need to be matched
based on strict preferences. Gale and Shapley famously proved that \emph{every}
instance of \SM\ indeed has a stable matching, and that one such matching can be
found in linear time~\cite{doi:10.1080/00029890.1962.11989827}. The
Gale-Shapley algorithm for \SM follows a simple---almost simplistic---greedy
strategy: in turn, each unmatched man proposes to the most preferred woman who
has not rejected him so far, and each woman holds on to the best proposal (as
per her preference) that she has got so far. Gale and Shapley proved that this
algorithm invariably finds a stable matching, which is said to be a
\emph{man-optimal} stable matching. Of course, the algorithm also works if the
women do the proposing; a stable matching found this way is said to be
\emph{woman-optimal}.

It is not difficult to come up with instances of \SM where the man-optimal and
women-optimal stable matchings are identical, as also instances where they
differ. A rich theory about the combinatorial structure of stable matchings has
been developed over the years. In particular, it is known that the set of all
stable matchings of a \SM instance forms a \emph{distributive lattice} under a
certain natural partial order, and that the woman-optimal and man-optimal stable
matchings form the maximum and minimum elements of this lattice. It follows that
each instance has exactly one man-optimal stable matching and one woman-optimal
stable matching, and that if these two matchings are identical, then the
instance has exactly one stable matching in total.

The Gale-Shapley algorithm can thus do a restricted form of \emph{counting}
stable matchings: it can correctly report that an instance has exactly one
stable matching, or that it has \emph{at least} two, in which case it can output
two different stable matchings. 
Our focus in this work is on finding a large collection of \emph{pairwise
  edge-disjoint} stable matchings:

\defproblem{\DSM}%
{A \SM instance \(G\) and an integer \(k\). }%
{Decide if \(G\) has at least \(k\) pairwise disjoint stable matchings, and output such a collection of stable matchings if it exists.}

Finding such a collection of disjoint stable matchings is clearly useful in
situations which involve \emph{repeated} assignments. For instance, when
assigning people to tasks---drivers to bus routes, medical professionals to
wards, cleaning staff to locations---this helps in avoiding monotony without
losing stability. As another example, consider a business school program which
has a series of projects on which the students are supposed to work in teams of
two. Using a different stable matching from a disjoint collection to pair up
students for each project will help with their collaborative skills while still
avoiding problems of instability.

Even in those cases where only one stable matching suffices---such as when
assigning medical students to hospitals once a year---a disjoint collection can
still be very useful. Given such a collection, an administrator in charge of
deciding the residencies can evaluate each stable matching based on other
relevant considerations---such as gender or racial diversity, or costs of
relocation---to choose an assignment which optimizes these other factors while
still being stable.

Our main result is that \DSM can be solved in \emph{linear} time:

\begin{theorem}
  \label{thm:mainResult}
  There is an algorithm which takes an instance \(G\) of \SM
  , runs in time linear in the size of the input, and outputs a
  pairwise disjoint collection of stable matchings of \(G\) of the largest size.
\end{theorem}
This immediately yields:
\begin{corollary}
  \label{cor:mainResult}
  \DSM can be solved in linear time.
\end{corollary}

Finding disjoint \emph{matchings} (without the
stability requirement) has received a lot of attention over the years, and a
number of structural and algorithmic results are
known~\cite{mkrtchyan2008edge,faudree1991neighborhood,10.2307/90020380}; we
mention just one, for perfect matchings in bipartite
graphs. 

Observe that a bipartite graph \(G\) has a perfect matching only if both sides
have the same size, say \(n\). Also, any collection of pairwise disjoint perfect
matchings of such a graph \(G\) can have size at most \(n\). This is because
deleting the edges of one perfect matching from \(G\) decrements the degree of
each vertex by exactly one, and the maximum degree of \(G\) is not more than
\(n\). A graph is said to be \(k\)-regular if each of its vertices has degree
exactly \(k\). K\"{o}nig proved that a bipartite graph \(G\) contains \(k\)
pairwise edge-disjoint perfect matchings if and only if \(G\) has a
\(k\)-regular subgraph~\cite{konig1916graphen}. Tutte's polynomial-time
algorithm for finding the so-called \(f\)-factors~\cite{tutte1954short} can be
used to find a \(k\)-regular subgraph of \(G\). Putting these together we get
a polynomial-time algorithm for finding a largest collection of edge-disjoint
perfect matchings in bipartite graphs.

In stark contrast, checking if a non-bipartite graph has \emph{two} disjoint
perfect matchings is already \textsf{NP}-hard even in \(3\)-regular
graphs~\cite{holyer1981np,DBLP:journals/corr/abs-2009-04567}. 
\textcolor{black}{We show an intractability result for the {\em stable roommates} problem.
\begin{theorem}\label{thm:roommates}
    Given a stable roommates instance, it is NP-complete to determine whether the instance admits at most two disjoint stable matchings or three disjoint stable matchings.
\end{theorem}
}

\subsection{Related Work and Subsequent Progress}

\textcolor{black}{There is a vast literature on the \SM\ and \SR problems. We refer the reader to the book \cite{Gusfield1989} for details. The \SM\ and \SR\ problems are known to be solvable in polynomial time due to \cite{doi:10.1080/00029890.1962.11989827} and \cite{Irving1985} respectively. Some structural properties of \SR, alongwith a polynomial-delay enumeration algorithm were given by \cite{Gusfield1988}. The \emph{maximum} number of stable matchings
that an instance can have has also received quite a bit of attention. Irving and
Leather~\cite{doi:10.1137/0215048} discovered a method for constructing
instances with exponentially-many stable matchings; these instances with \(n\)
men and \(n\) women have \(\Omega(2.28^{n})\) stable matchings. This is the
current best lower bound on the maximum number of stable matchings. The upper bounds are 
\(\mathcal{O}(c^{n})\) for some constant
\(c\)~\cite{10.1145/3188745.3188848,palmer2020447n} and $O(3.55^n)$ by \cite{PalmerP2021}.}

\paragraph*{Relation to lattice structure:} It is known that the set of stable
matchings in a given \SM\ instance forms a distributive lattice \cite{Knuth76}. We
show that there is always a solution to \DSM that is a chain in this lattice. We
give an algorithm to enumerate all the chains of disjoint stable matchings. The
algorithm takes time polynomial in the size of the input for outputting each
such chain. We also show that the expected number of such chains in a random
instance is at most quasi-polynomial with high probability.

\paragraph*{Subsequent developments:}
\textcolor{black}{After the conference version of this work was published, there have been subsequent developments on this problem. \cite{deBergOM2025} give a general framework for diverse solutions, including an efficient polynomial-time algorithm to find a set of stable matchings of desired size that maximize the sum of pairwise hamming distances. 
Another work is by \cite{FleinerFK2024Approach} that gives a structural min-max theorem for stable matchings relating the maximum number of edge-disjoint stable matchings and the minimum number of stable matchings required to cover all stable edges.
}

\section{Preliminaries}\label{sec:prelim}

We recall the Gale-Shapley algorithm and the lattice structure of stable
matchings here for the sake of completeness. The classical Gale-Shapley
algorithm~\cite[Figure~1.3]{Gusfield1989} solves the \SM problem by a deferred
acceptance mechanism. Each man proposes the women on his list in decreasing
order of preference until some woman accepts his proposal. A woman $w$ accepts a
proposal from a man $m$ if either $w$ is unmatched or she prefers $m$ over her
current partner. The extended version of the Gale-Shapley algorithm
(\autoref{algorithm:extended-gale-shapley})~\cite[Figure~1.7]{Gusfield1989}
\textit{reduces} the preference lists by eliminating certain pairs that do not
belong to any stable matching. By \textit{deleting} a (\textit{man-woman}) pair
$(m,w)$, we mean deleting $m$ from $w$'s preference list and $w$ from that of
$m$. The formal pseudocode is given in Appendix.


 %

The algorithm terminates when every man is engaged or has exhausted his
preference list. When the algorithm ends, the resulting modified preference list
is a \textit{reduced list}. Furthermore, it can be easily verified that, on
termination, each man is either unmatched or is engaged to the first woman in
his \textit{reduced} preference list, and each woman is either unmatched or is
engaged to the last man in hers. These engaged pairs constitute a man-optimal
stable
matching. 
It is known that every stable matching leaves the \emph{same} set of people
unmatched~\cite{Gusfield1989}.

For a given stable marriage instance we will refer to the final preference lists
generated by 
\EGS, with men as proposers, as
\textit{man-oriented Gale-Shapley lists}, or \textit{MGS-lists}. The final
preference lists generated by this algorithm when women do the proposing are
called \textit{WGS-lists}. Finally, if we take for each person the intersection
of their MGS-list and WGS-list, we get the \textit{GS-list}. It is known that
the GS-lists can be obtained by first applying man-oriented 
\EGS to get MGS-lists and then, starting with the MGS-lists, applying
woman-oriented 
\EGS~\cite{Gusfield1989}.

Let $G_{GS-list}$ be the graph obtained from the GS-lists as follows: Each man
$m_i$ is represented by a vertex $m_i$ and each woman $w_i$ is represented by a
vertex $w_i$, and an edge $(m_i,w_i)$ is present if and only if $m_i$ is in
$w_i$'s preference list in the GS-lists. We say that a matching $M$ is
\emph{contained} in the GS-lists if $M$ is a subgraph of $G_{GS-list}$.

The next theorem summarizes some useful properties of GS-lists.
\begin{theorem}\emph{\cite[Theorem~1.2.5]{Gusfield1989}}\label{thm:gs-lists}
  For a given instance of the stable marriage problem,
  \begin{enumerate}
  \item all stable matchings are contained in the GS-lists;
  \item no matching (stable or otherwise) contained in the GS-lists can be blocked by a pair that is not in the GS-lists;
  \item In the man-optimal (respectively woman-optimal) stable matching, each man is partnered by the first (respectively last) woman on his GS-list, and each woman by the last (respectively first) man on hers. 
  \end{enumerate}
\end{theorem}

\paragraph*{Lattice structure of stable matchings.} We need the following
results about the lattice structure of stable matchings~\cite{Gusfield1989}. For
a given stable marriage instance, a \emph{dominance relation} on stable
matchings is defined as follows:
\begin{definition}[Dominance]\label{def:dominance}
  A stable matching $M$ is said to \emph{dominate} a stable matching $M'$,
  written $M \preceq M'$, if every man has at least as good a partner in $M$ as
  he has in $M'$; i.e., every man either prefers $M$ to $M'$ or is indifferent
  between them.
\end{definition}

\begin{lemma}\emph{\cite[Lemma~1.3.1]{Gusfield1989}}
  For a given stable marriage instance, let $M$ and $M'$ be two (distinct) stable matchings. If each man is given the better of his partners in $M$ and $M'$ (denoted as $M \wedge M'$), then the result is a stable matching that dominates both $M$ and $M'$. 
\end{lemma}

\begin{lemma}\emph{\cite[Lemma~1.3.2]{Gusfield1989}}
  For a given stable marriage instance, let $M$ and $M'$ be two (distinct) stable matchings. If each man is given the poorer of his partners in $M$ and $M'$ (denoted as $M \vee M'$), then the result is a stable matching that is dominated by both $M$ and $M'$. 
\end{lemma}

With the help of the above lemmas, it is easy to see that the set of all stable matchings forms a distributive lattice and the man-optimal matching and the woman-optimal matching represent the minimum and maximum elements of the lattice  \cite[Theorem~1.3.2]{Gusfield1989}. Moreover, $M \wedge M'$ represents the \emph{greatest lower bound} and $M \vee M'$ represents \emph{least upper bound} of $M$ and $M'$ in the lattice of all the stable matchings.

\section{Finding Disjoint Stable Matchings}\label{sec:algorithm}

In this section we describe and analyze our algorithm for finding a largest
collection of disjoint stable matchings in a given instance of \SM.

Given a stable marriage instance, two matchings \(M_1\) and \(M_2\) are said to
be {\em disjoint stable matchings} if both \(M_1\) and \(M_2\) are stable and
they do not share a common edge. Throughout this section, we denote the
man-optimal and woman-optimal stable matchings by $M_o$ and $M_z$ respectively.
The following lemma gives a necessary condition for the existence of two or more
disjoint stable matchings for a given marriage instance. An example is given in
Appendix~\ref{app-sec:ex}.
\begin{lemma}\emph{\cite[Section~1.2.2]{Gusfield1989}}
  Let $(m,w)$ be a pair in $M_o\cap M_z$. Then $(m,w)$ is contained in every stable matching. 
\end{lemma}

The algorithm first finds the man-optimal and woman-optimal stable matchings
($M_o$ and $M_z$ respectively) by executing
\EGS. 
If these matchings share an edge, the algorithm stops. Otherwise it modifies the
instance by deleting all the edges that appear in $M_o$. It then computes a
man-optimal matching \(M'\) of the new instance using \EGS. If \(M'\) is
disjoint from the woman-optimal matching \(M_{z}\) then it deletes the edges of
\(M'\) from the instance. The algorithm repeats this procedure as long as \EGS
keeps returning a stable matching which is disjoint from \(M_{z}\). It stores
all the \(M_{z}\)-disjoint matchings obtained during this process in a set $S$.
We note that this is a stronger version of the \emph{BreakMarriage} algorithm of
McVitie and Wilson~\cite{McVitie71}.

\begin{algorithm}[hbt]
  \caption{Disjoint Stable Matchings}\label{algo:one}
  \label{algorithm:disjoint-SM}
  \hspace*{\algorithmicindent} \textbf{Input \ :} A stable matching instance \(G\) \\
  \hspace*{\algorithmicindent} \textbf{Output:} A maximum size set \(S\) of disjoint stable matchings.
  \begin{algorithmic}[1]
    \Procedure {Disjoint stable matchings}{$G$}
    \State \texttt{$S \gets \varnothing$}
    \State \texttt{$M_z \gets$}\Call{StableMatching}{$G$, woman-optimal} \Comment{\scriptsize Woman-proposing GS Algorithm}\normalsize \label{line:stm}
    \If {$(M_z =\varnothing)$} \label{line:check}
    
    \Return {$\varnothing$}
    \EndIf
    \State \texttt{$X \gets $}\Call{GS-Extended}{$G$} \Comment {\scriptsize This modifies  preference lists}\normalsize\label{line:egs1}
    \While {\texttt{$X \cap M_z = \varnothing$}}\label{line:while}
    \State \texttt{$S \gets S \cup \{X\}$}
    \For {\texttt{every man} $m$}\label{line:for-loop}
    \State \texttt{Delete the first woman $w$ on $m$'s list} \Comment {\scriptsize \(m\)'s partner in \(X\)}\normalsize \label{line:delete-w}
    \State \texttt{Delete the last man on $w$'s list} \Comment {\scriptsize \(w\)'s partner in \(X\)}\normalsize \label{line:delete-m}
    \EndFor
    \State \texttt{$X \gets
      $}\Call{\hyperref[algorithm:extended-gale-shapley]{GS-Extended}}{$G$} \Comment {\scriptsize Get a new disjoint matching as $X$}\normalsize \label{line:egs2}  
    \EndWhile
    \State \texttt{$S \gets S \cup \{X\}$}
    
    \Return \texttt{$S$}
    \EndProcedure    
  \end{algorithmic}

\end{algorithm}

We first show that the matchings in the set $S$ constructed by
\autoref{algo:one} are stable. They are clearly disjoint by construction,
since each step starts off by deleting \emph{every} matched pair in the
matching computed in the previous step. 
\begin{lemma}\label{lem:stable}
  All the matchings in the set $S$ are stable matchings.
\end{lemma}
\begin{proof}
  For the sake of contradiction, let $(m,w)$ be a blocking pair for a matching
  $M_i \in S$. Then, $m$ prefers $w$ to $p_{M_i}(m)$, where $p_{M_i}(m)$ is the
  partner of $m$ in $M_i$. That is, $w$ appears before $p_{M_i}(m)$ in $m$'s
  preference list. As $m$ is matched to $p_{M_i}(m)$ in the matching $M_i$, $w$
  would have been deleted from $m$'s preference list \textit{before} the call to
  GS-Extended that returned the matching $M_i$. This deletion can happen in two
  ways. Either in one of the calls to the Extended GS algorithm, or in one of
  the iterations of the for loop in line~\ref{line:for-loop} of the algorithm.
  We know that in both the cases, after the deletion of $w$ from $m$'s
  preference list, $w$ gets a strictly better partner than $m$ in the subsequent
  matching. Therefore, $w$ does not prefer $m$ to $p_{M_i}(w)$. This contradicts
  our assumption.\qed
\end{proof}
Building on the notion of dominance from \autoref{def:dominance}, we say that
$M$ {\em strictly dominates} $M'$, denoted by $M\prec M'$, if $M\preceq M'$ and
$M\cap M'=\emptyset$.
The strict dominance relation imposes a partial order on the set of stable
matchings in $G$. We call a set of stable matchings a \emph{chain} if it forms a
chain under the (non-strict) dominance relation of \autoref{def:dominance}. Let
$M_i$ be the matching included in $S$ at the end of iteration $i$ of the
algorithm, and let $|S|=k$.
\begin{lemma}\label{lemma:chain}
The stable matchings in the set $S$ form a chain $M_o=M_1,\ldots,M_k$.
\end{lemma}
\begin{proof}
Each iteration of the algorithm modifies the given instance by deleting 
the edges of the matching constructed. Let the instance considered at the beginning of iteration $i$ be $G_i$.
Thus $G_1=G$. Since $M_i$ is constructed by executing the extended Gale-Shapley algorithm on the instance $G_i$,
it follows that $M_i$ is the man-optimal matching in $G_i$.
Further, all the men get strictly
better partners in $M_i$ compared to $M_j$, $j>i$ and all the women get strictly worse partners in $M_i$
compared to $M_j$, for $j>i$.\qed
\end{proof}

We now show that among all the chains of disjoint stable matchings, the one output by \autoref{algorithm:disjoint-SM} is a 
longest chain. 
\begin{lemma}
  \label{lemma:longest-chain}
  \autoref{algorithm:disjoint-SM} outputs a longest chain of disjoint stable matchings.
\end{lemma}
\begin{proof}
  Let $M_o=M_1\prec M_2\prec \cdots \prec M_k$ be the chain of disjoint matchings obtained by 
running \autoref{algorithm:disjoint-SM}. For the sake of contradiction, let $M'_1\prec M'_1\prec \cdots \prec M'_\ell$ be a longest 
chain of disjoint matchings such that $\ell>k$.

  We know that the matching $M_1=M_o$ dominates \emph{every} stable matching~\cite[Theorem~1.2.2]{Gusfield1989}. 
Matching $M'_1$ cannot be disjoint with $M_1$, as otherwise, $M_1\prec M'_1\prec M'_2\prec \cdots \prec M'_{\ell}$ 
would be a longer chain of disjoint stable matchings. Therefore, $M'_1$ shares some edges with $M_1$. As $M_1 \preceq M'_1 \prec M'_2$, we have 
$M_1 \prec M'_2$. Therefore we can replace $M'_1$ in $M'_1\prec M'_2\prec \cdots \prec M'_\ell$ with $M_1$ to get another chain of disjoint 
stable matchings $M_1\prec M'_2\prec \cdots \prec M'_\ell$ of length $\ell$.

  We know that $M_2$ dominates all the stable matchings which are disjoint with $M_1$. Matching $M'_2$ cannot be  disjoint with $M_2$, as otherwise, we can get 
a longer chain $M_1 \prec M_2 \prec M'_2 \prec M'_3 \prec \cdots \prec M'_\ell$. Therefore, $M'_2$ shares edges with $M_2$. As $M_2 \preceq M'_2 \prec M'_3$, 
we have $M_2 \prec M'_3$. Therefore we can replace $M'_2$ with $M_2$ to get another chain of disjoint stable matchings 
$M_1\prec M_2\prec M'_3 \prec \cdots \prec M'_\ell$ of length $\ell$.


In this way, we successively replace each $M'_i$ with $M_i$ to get the
$\ell$-length chain \(M_1 \prec M_2 \prec \dotsc M_{k} \prec M'_{k+1} \prec \cdots \prec
M'_\ell\) of disjoint stable matchings. But this implies that there exists a
stable matching \(M'_{k+1}\) which satisfies the strict relation \(M_{k} \prec
M'_{k+1}\), which is a contradiction since \(M_{k}=M_{z}\) is the
woman-optimal stable matching.\qed 
\end{proof}

We have shown that among all the chains of disjoint stable matchings, the one
output by \autoref{algorithm:disjoint-SM} is of maximum length. We still need to
prove that there is no larger set of disjoint stable matchings which is possibly
not a chain. We use the following result due to Teo and Sethuraman to show that
any such set of disjoint stable matchings has a corresponding chain of disjoint
stable matchings. Moreover, the length of this chain is same as the size of the
set.
\begin{theorem}\emph{\cite{teo1998geometry}}
  \label{theorem:teo-sethuraman}
  Let $S=\{M_1,M_2,\cdots,M_k\}$ be a set of stable matchings for a particular
  stable matchings instance. For each man $m$, let $S_m$ be the \emph{sorted}
  multiset \(\{p_{M_1}(m),p_{M_2}(m),\cdots,p_{M_k}(m)\}\), sorted according to
  the preference order of \(m\). For every $i\in\{1,2,\cdots,k\}$ let 
  \( M_i' = \{(m,w) \mid m \in\Men \text{ and }$w$ \text{ is the }\\i^{\text{th}}\text{ woman in }S_m\}\).
  Then for each $i \in \{1,2,\cdots,k\}$, $M_i'$ is a stable matching.
\end{theorem}
The following is an immediate corollary of \autoref{theorem:teo-sethuraman}:
\begin{corollary}
  \label{corollary:disjoint-chain}
  Let $M_1,\ldots,M_k$ and $M'_1,\ldots, M'_k$ be as defined in \autoref{theorem:teo-sethuraman}. If $M_1,\ldots, M_k$ are
  pairwise disjoint, then $M'_1,\ldots,M'_k$ form a $k$-length chain of disjoint stable matchings.

\end{corollary}

The following theorem now completes the correctness of \autoref{algorithm:disjoint-SM}.

\begin{theorem}
  \label{theorem:disjoint-SM}
  For a given stable marriage instance, \autoref{algorithm:disjoint-SM} gives a maximum size set of disjoint stable matchings.
\end{theorem}
\begin{proof}
  Let $S= \{M_1=M_o,M_2,\cdots,M_k\}$ be the set of disjoint stable matchings output 
by \autoref{algorithm:disjoint-SM}. For the sake of contradiction, Let $S' = \{M'_1,M'_2,\cdots, M'_\ell\}$ 
be a maximum size set of disjoint stable matchings such that $\ell>k$. Then, from \autoref{corollary:disjoint-chain} 
of \autoref{theorem:teo-sethuraman}, we know that there exists an \(\ell\)-length chain of disjoint 
stable matchings. This contradicts \autoref{lemma:longest-chain}, that the $k<\ell$ matchings from $S$ form a \emph{longest} chain of disjoint stable matchings. 
\end{proof}
\paragraph*{Time complexity:} Each edge of $G$ is visited exactly once during the course of the algorithm. Hence the time complexity is $O(m+n)$ where
$2n$ is the number of vertices in $G$ and $m$ is the number of edges in $G$.
This completes the proof of \autoref{thm:mainResult}.


\section{Enumerating all max-length Chains}
Algorithm~\ref{algo:one} gives one maximum-length chain of disjoint stable matchings. It is an interesting question whether such a chain
is unique. Example~\ref{fig:marriage_instance_4.1} from Appendix shows that there can be multiple maximum-length chains of disjoint stable 
matchings. We now give an algorithm to enumerate all such chains with polynomial delay. For the enumeration, we exploit the lattice structure
of stable matchings described in Section~\ref{sec:prelim}.

The $\#P$-hardness of counting all the maximum-length chains can be easily deduced from the $\#P$-hardness of counting all the stable
matchings in a given instance. If we add a new man-woman pair $(m,w)$ such that both prefer each other over all the others, then every
stable matching in the new instance contains the pair $(m,w)$. Hence the length of a maximum-length chain of disjoint stable matchings is $1$,
and each stable matching in the given instance is such a chain.


Let $A_0$ be the man optimal matching. Define the set $A = \{A_0, A_1, \dots A_k\}$ such that for $1 \leq i \leq k$, $A_i = \bigvee \{M| A_{i-1}\prec M, M \cap A_{i-1} = \varnothing\}$
Similarly, let $B_0$ be the woman optimal stable matching. Define the set $B = \{B_0, B_1, \dots, B_t\}$ such that for $1 \leq i \leq t$, $B_i=\bigwedge \{M| B_{i-1}\succ M, M \cap B_{i-1} = \varnothing\}$\\
We note that $A$ and $B$ are the chains returned by Algorithm~\ref{algorithm:disjoint-SM} with man-proposing and woman-proposing versions respectively. Since both are maximum-length chains of disjoint stable matchings, $t = k$.

Let $X = \{X_0, \cdots  X_k\}$ be a maximum-length chain of disjoint stable matchings i.e. $X_0 \prec X_1 \prec \dots \prec X_k$.
We note the following property of the matchings in $X$.

\begin{lemma}
  \label{lemma:chains}
  For $0 \leq i \leq k$, $A_i \preceq X_i \preceq B_{k-i}$
\end{lemma}
\begin{proof}
  By induction on $i$, we prove $A_i \preceq X_i$ for $0 \leq i \leq k$. Proving $X_i \preceq B_{k-i}$ is analogous.
  As $A_0$ is the man-optimal matching, $A_0 \preceq X_0$.
  Assume for some $i$, $A_i \preceq X_i$. Hence $A_i \preceq X_i \prec X_{i+1}$. 
  Therefore $A_i$ and $X_{i+1}$ are disjoint. Since $A_i \prec X_{i+1}$, $X_{i+1}$ is disjoint with $\{A_0, \dots, A_{i-1}\}$ too.
  By definition of $A$, $A_{i+1}$ is the greatest lower bound of all matchings that do not intersect with $\{A_0, \dots, A_i\}$. Thus, $A_{i+1} \preceq X_{i+1}$.
\end{proof}

\begin{corollary}\label{cor:prec}
For each $i$, $A_i \preceq B_{k-i}$. Moreover, $\{A_0, \dots, A_{i-1}, X_i, B_{k-i-1}, \dots, B_0\}$ is also a maximum chain of disjoint stable matchings.
\end{corollary}
We now present the enumeration algorithm.\\
{\em Outline: }Algorithm~\ref{algorithm:enumerate} constructs the sub-lattice $L_i$ of stable matchings between $A_i$ and $B_{k-i}$ for each $i$. This
is described in detail later.
For each matching $M$ in $L_i$, it computes the greatest lower bound $U$ of the matchings that strictly dominate $M$. The computation of $U$ is
done in the same spirit as Algorithm~\ref{algorithm:disjoint-SM}. To find $U$, we create a 
new instance by deleting from the preference list of each man $m$ (respectively each woman $w$), the women that $m$ prefers at least as much as $p_M(m)$
(respectively the men that $w$ prefers at most as much as $p_M(w)$). From Lemma~\ref{lemma:chain}, the required matching $U$ is the man-optimal
stable matching in the new instance described above. The algorithm marks such matchings, and makes a list of matchings $V$ such that $V$ is in the same sub-lattice as $U$ and $U\preceq V$. 

The above information is used for creating a directed graph $H$. The vertices of $H$ are all the stable matchings in the given instance $G$. The edges
incident on a matching $M$ are $(M,V)$ where $V$ is as described above. A source $s$ and a sink $t$ are added to $H$. The edges from $s$ are incident on all the matchings in $L_0$, and the edges to $t$ are from all the matchings in $L_k$.

\begin{algorithm}[!hbt]
  \caption{Enumeration}
  \label{algorithm:enumerate}
    \hspace*{\algorithmicindent} \textbf{Input \ :} A stable matching instance \(G\), the output of man-oriented version of Algorithm $2$ $A = \{A_0, A_1, \dots, A_k\}$ and the output of women-oriented version of Algorithm $2$ $B = \{B_0, B_1, \dots, B_k\}$\\
  \hspace*{\algorithmicindent} \textbf{Output:} All maximum size chains of disjoint stable matchings in \(G\).

  \begin{algorithmic}[1]
    \For {$0 \leq i \leq k$}
    \State Construct the sub-lattice $L_i$ of all stable matchings $M$ such that $A_i \preceq M \preceq B_{k-i}$.
    \EndFor
    \For {$0 \leq i \leq k$}
    \For{every $M \in L_i$}
    \State Compute $U_M = \bigwedge S$ for $S=\{M'|M\prec M', M\cap M' =\varnothing\}$. Mark $U_M$. 
    \EndFor
    \EndFor
    \For {$0 \leq i \leq k$}
    \For{every marked $U \in L_i$}
    \State Create a list of all $M^\prime \in L_i$ such that $U \preceq M^\prime$ and attach the list $S[U]$ to $U$.
    \EndFor
    \EndFor
    \State Let $H$ be an empty graph.
    \For { each matching $M \in L_0 \cup \cdots \cup L_k$}
    \State Add a vertex $v_M$ in $H$.
    \EndFor
    \State Add vertices $s$ and $t$ in $H$
    \For {$0 \leq i < k$}
    \For{$M \in L_i$}
    
    \State Connect $(M, U_M)$ with an edge in $H$.
    \State Connect edge $(M, V)$ for all $V \in S[U_M]$ in $H$.
    \EndFor
    \EndFor
    \State Add edges $(s, M)$ for all $M \in L_0$.
    \State Add edges $(M, t)$ for all $M \in L_k$.
    \State Enumerate all $s$-$t$ paths in $H$. For each path $p$ in $H$, print the set of stable matchings corresponding to vertices in $p \setminus \{s, t\}$.
  \end{algorithmic}
\end{algorithm}

It follows from Lemma~\ref{lemma:chains} that every maximum chain of disjoint stable matchings corresponds to some $s-t$ path in $H$. We now establish the converse.

\begin{lemma}
  \label{lemma:converse}
  Every $s$-$t$ path in $H$ corresponds to a distinct maximum chain of disjoint stable matchings.
\end{lemma}

\begin{proof}
  For an $s$-$t$ path $p$, let $X = \{X_0, X_1, \dots, X_\ell\}$ ($\ell \leq k$) be the set of stable matchings corresponding to vertices in $p \setminus \{s, t\}$. For $0 \leq i < \ell$, there exists an $s$-$t$ path visiting both $X_i$ and $X_{i+1}$ since $U \preceq X_{i+1}$ and $X_{i+1} \in S[U]$, where $U = \bigwedge S$, $S=\{M'|X_i \prec M', X_i \cap M' =\varnothing\}$. Thus, $X_i \preceq U \preceq X_{i+1}$. Thus, the set $X$ forms a chain. Further, since $U$ and $X_i$ are disjoint by construction, $X_i$ and $X_{i+1}$ are also disjoint. Since $X$ is a chain satisfying $X_0 \preceq X_1 \preceq \dots \leq X_\ell$ and any two consecutive matchings are disjoint, the set $X$ forms a chain of disjoint stable matchings. It remains to prove
that it is a maximum-length chain i.e. $\ell=k$.

Suppose $\ell < k$. Then, there exists a smallest $i$ such that no matching in $X$ belongs to $L_i$. Note that $i \neq 0$, since any edge leaving $s$ enters $L_0$. Thus, there exists $X_j \in X$, $X_j \in L_{i-1}$.
Note that $X_{j} \succeq B_{k-i+1} \succeq B_{k-i}$. Further, since $B_{k-i+1}$ and $B_{k-i}$ are disjoint, $X_j$ and $B_{k-i}$ are also disjoint. Hence, the greatest lower bound $U$ of all matchings greater than $X_j$ and disjoint to $X_j$ is greater than $B_{k-i}$. Also, note that $A_{i-1}$ and $B_{k-i+1}$ intersect. Thus, any two matchings in $L_{i-1}$ intersect. In particular, $U$ cannot belong to $L_{i-1}$, since then $U$ would intersect with $X_j$. Thus, $U \succeq B_{k-i}$, but $U \not \in L_{i-1}$. Thus, $U \in L_i$.
  Since $U \in L_i$, by construction, all edges leaving $X_j$ enter $L_i$, and thus $p$ has a vertex in $L_i$ and $X$ has a matching in $L_i$. 
  
  Thus, for all $0 \leq i \leq k$, there exists a matching $X_i \in L_i$ in $X$. Thus, $|X| = k$ and $X$ is a maximum-length chain of disjoint stable matchings. The chain corresponding to an $s$-$t$ path is essentially the set of vertices in the path (excluding $s$ and $t$). Thus, the corresponding chains of two such paths cannot be equal. 
\end{proof}

\paragraph*{Constructing the sub-lattice $L_i$ between $A_i \prec B_{k-1}$: }
We begin by construct an updated list $P$ of preferences as follows:
\begin{enumerate}
\item Delete every woman in $m$'s list better than his partner in $A_i$ and worse than his partner in $B_{k-i}$. Delete every man in $w$'s list better than her partner in $B_{k-i}$ and worse than her partner in $A_i$.
\item Update the preference list so that $m$ is in $w$'s list iff $w$ is in $m$'s list.  
\end{enumerate}
Constructing $P$ takes $O(n^2)$ time. Observe that $A_i$ and $B_{k-i}$ are the man-optimal and the woman-optimal stable matchings respectively in $P$. Also, $M$ is a stable matching in $P$ iff $M$ is a stable matching in the original instance. Moreover, $A_i \leq M \leq B_{k-i}$. We then use the algorithm from \cite{Gusfield87}, to enumerate all stable matchings in the instance $P$ and construct the Hasse diagram for the lattice of stable matchings in $P$ in $O(n^2 + n|T|)$ time. Here, $T$ is the set of stable matchings in $P$.




\begin{lemma}\label{lemma:time-enum}
  Algorithm~\ref{algorithm:enumerate} terminates in $O(n^3 + n^2(|L| + |P|))$ time, where $P$ is the set of maximum-length chains of disjoint stable matchings.
\end{lemma}



\begin{proof}
  Constructing each of the sub-lattices takes $O(n^2 + n|L_i|)$ time. Thus, constructing all sub-lattices takes time proportional to $O(n^2|L|)$. 
  Computing $U$ given $M \in L_i$ takes $O(n^2)$ time. 
The list $S[U]$ can be computed in $O(n^2|S[U]|)$ time. The graph $H$ has $|L|$ vertices. Let $E_l$ be the number of edges in $H$. For each edge $(M, V)$ for $M \in L_i$, $V \in L_{i+1}$ and $(M, V) \not \in \{(A_i, A_{i+1}), (B_{k-i}, B_{k-i-1})\}$, mapping $(M, V)$ to the $s-t$ path $(s, A_0, \dots, A_{i-1}, M, V, B_{k-i-2}, \dots, B_0, t)$ is injective. Thus, $|E_L| \leq |P| + 2n$.\\
The maximum size of a chain of disjoint stable matchings is $n$. Hence enumerating each path can
be done in time $O(n)$ time. The total time complexity is thus $O(n^2(|L|+|P| + |E_L|)) = O(n^3 + n^2(|L| + |P|))$. In particular, after the construction of $H$, the time to
output each $s$ to $t$ path is polynomial in $n$.



\end{proof}

We analyze the number of maximum-length chains of disjoint stable matchings in a random stable matchings instance with complete lists.
Given a natural number $n$, we create a random stable matchings instance of $n$ men and $n$ women by assigning any of the $n!$ possible preference lists to each men and women uniformly at random. 

\begin{lemma}\label{lem:expected}
The probability of the number of maximum size chains of disjoint stable matchings exceeding $(\frac{n}{\ln n})^{\ln n}$ is at most $O(\frac{(\ln n)^2}{n^2})$.
\end{lemma}
\begin{proof}
  Let $S$ be the random variable denoting the number of stable matchings in a random stable matching instance. Pittel~\cite{Pittel89} showed that 
  $\mathbb{E}[S] = \Theta(n \ln n)$. Thus, there exist non-negative reals $m_1, m_2$ such that $m_1 n \ln n \leq \mathbb{E}[S] \leq m_2 n \ln n$ for sufficiently large $n$. Further, Lennon and Pittel~\cite{LP09} established that $Var(S) = \sigma^2 = O((n \ln n)^2)$. Thus, for sufficiently large $n$, there exists a non-negative real number $c$ such that $Var(S) \leq c^2(n \ln n)^2$.

  Thus, for a parameter $k$, we have 
  \begin{eqnarray*}
  Pr(S \geq m_1 n \ln n + kcn \ln n) & \leq & Pr(S \geq m_1 n \ln n + kcn \ln n \cup S \leq m_2 n \ln n - kcn \ln n) \\
  & \leq & Pr(|S - \mathbb{E}[S]| \geq kcn \ln n) \\
  & \leq & Pr(|S - \mathbb{E}[S]| \geq k\sigma)\\ 
  & \leq & \frac{1}{k^2}
  \end{eqnarray*}
   where the last inequality follows from Chebyshev's inequality. Thus, if $f(k) = m_1 n \ln n + kcn \ln n$, then $Pr(S \geq f(k)) \leq \frac{1}{k^2}$.

  Let $L_0, L_1, \dots, L_{t-1}$ be the sub-lattices constructed in Algorithm~\ref{algorithm:enumerate} where $t-1=k$. 
  Let $S_i = |L_i|$ for $0 \leq i \leq k$. Let $p = |P|$,
  the number of maximum-length chains of disjoint stable matchings in the given instance. From Lemma~\ref{lemma:converse}, we have 
  $p \leq \Pi_{i=0}^k S_i \leq (\frac{\sum_{i=0}^t S_i}{t})^t$, where the last inequality follows from the AM-GM inequality. Since $\sum_{i = 0}^k S_i \leq S$, $p \leq (\frac{S}{t})^t$.

  From the above discussions, $Pr(p \geq (\frac{n}{\ln n})^{\ln n})\leq Pr((\frac{S}{t})^t \geq (\frac{n}{\ln n})^{\ln n}) \leq Pr(S \geq n^2) + Pr(t \geq \ln n)$.\\

  Observe that there exists a positive real $m$ such that $f(\frac{n}{m \ln n}) \leq n^2$. Thus, $Pr(S \geq n^2) \leq Pr(S \geq f(\frac{n}{m \ln n})) \leq \frac{m^2 (\ln n)^2}{n^2}$. [Knuth et al $90$] establishes that the probability of some person having more than $\ln n$ stable partners is super-polynomially small. Clearly, no one can have less than $t$ stable partners since each person features alongide a distinct partner in each matching in a maximum size chain of disjoint stable matchings. Hence, $Pr(t \geq \ln n)$ is also super-polynomially small.\\

  Thus, $Pr(p \geq (\frac{n}{\ln n})^{\ln n}) \leq \frac{m_1^2 (\ln n)^2}{n^2}$ for some positive constant $m_1$. Thus, $Pr(p \geq (\frac{n}{\ln n})^{\ln n}) \leq O(\frac{(\ln n)^2}{n^2})$. $\qed$
\end{proof}

\begin{corollary}
  Algorithm~\ref{algorithm:enumerate} terminates in $O(n^4 + n^{2 \ln n + 2})$ time with probability $1$ as $n \xrightarrow{} \infty$.
\end{corollary}

\begin{proof}
  As established in the previous lemma (notation carrying over from the proof of the previou lemma), $Pr(S \geq n^2) \leq O(\frac{(\ln n)^2}{n^2})$ and $Pr(p \geq (\frac{n}{\ln n})^{\ln n}) \leq O(\frac{(\ln n)^2}{n^2})$ and hence, a simple union bound returns $Pr(S \geq n^2 \cup p \geq (\frac{n}{\ln n})^{\ln n}) \leq O(\frac{(\ln n)^2}{n^2})$.\\

  Plugging in $S = O(n^2)$ and $p = O(\frac{n}{\ln n})^{\ln n})$ in the run-time of algorithm $1$, algorithm $1$ terminates in $O(n^4 + n^{2 \ln n + 2})$ time with probability $1 - \Omega(\frac{(\ln n)^2}{n^2})$ which tends to $1$ as $n \xrightarrow{} \infty$. $\qed$
\end{proof}

\section{NP-hardness in Stable Roommates Instance}\label{sec:roommates}
In this section, we prove Theorem~\ref{thm:roommates}. 

We show that determining if an instance of the \SR\ problem admits three disjoint stable matchings is NP-complete. We give a polynomial-time reduction from the \threecol\ problem in graphs, known to be NP-complete due to \cite{GareyJS1976}. Let $G=(V,E)$ be a given graph. We create a stable roommates instance $H$ which consists of a set of agents $\mathcal{A}$, and a set of preference lists $\{P_a\}_{a\in A}$ where $P_a$ is a permutation on $\mathcal{A}\setminus \{a\}$. Thus, every agent in $\mathcal{A}$ has a strict and complete preference list over other agents in $\mathcal{A}$.

\noindent{\em The construction:} The set of agents $\mathcal{A}$ is given by $\mathcal{A}=\{a^1_v,a^2_v,a^3_v,b^1_v,b^2_v,b^3_v\mid v\in V\}$, i.e. corresponding to each vertex in $G$, there are six agents in $\mathcal{A}$. The preference lists are given below.

\begin{eqnarray*}
    P_{a^1_v} & : & b^1_v \succ b^2_v\succ \{a^1_u\mid (u,v)\in E\}\succ b^3_v\succ \textrm{all remaining agents}\\
    P_{a^2_v} & : & b^2_v\succ b^3_v\succ b^1_v \succ \textrm{all remaining agents}\\
    P_{a^3_v} & : & b^3_v\succ b^1_v\succ b^2_v\succ \textrm{all remaining agents}\\
    P_{b^1_v} & : & a^2_v\succ a^3_v\succ a^1_v\succ \textrm{all remaining agents}\\
    P_{b^2_v} & : & a^3_v\succ a^1_v\succ a^2_v \succ \textrm{all remaining agents}\\
    P_{b^3_v} & : & a^1_v\succ a^2_v\succ a^3_v \succ \textrm{all remaining agents}
\end{eqnarray*}
A subset of agents, where a specific order is not mentioned, can be ordered arbitrarily. Thus the preference lists are strict and complete.
We refer to the set $\{a^1_v,a^2_v,a^3_v,b^1_v,b^2_v,b^3_v\}$ as the {\em gadget} corresponding to a $v\in V$. 
Also, we denote $A_v=\{a^1_v,a^2_v,a^3_v\}$ and $B_v=\{b^1_v,b^2_v,b^3_v\}$. 
\begin{proposition}\label{prop:roommates}
Any stable matching in $H$ matches each agent in $B_v$ to an agent in $A_v$. Moreover, each gadget has exactly three stable matchings.\footnote{A matching is considered stable in a gadget if it is stable in the instance restricted to the gadget.}
\end{proposition}
\begin{proof}
Consider a matching $M$ that matches some agent $b^i_v$, $i\in [3]$ to an agent outside $A_v$. But $b^i_v$ is the first choice of $a^i_v$, $i\in [3]$. Thus $(a^i_v,b^i_v)$ blocks $M$.
    
    Now, consider the three matchings given by 
    \begin{eqnarray*}
        M^1_v &= & \{(a^1_v,b^1_v),(a^2_v,b^2_v),(a^3_v,b^3_v)\}\\
        M^2_v & = & \{(a^1_v,b^3_v),(a^2_v, b^1_v), (a^3_v,b^2_v) \}\\
        M^3_v & = & \{(a^1_v,b^2_v),(a^2_v,b^3_v),(a^3_v,b^1_v)\}
    \end{eqnarray*}
    The matching $M^1_v$ is stable within $A_v\cup B_v$ because each $a^i_v$ is matched to its top choice, and no two agents in $B_v$ can form a blocking pair because each of them prefers any vertex in $A_v$ over any other agent in $B_v$. Similarly, $M^2_v$ is also stable. To see that $M^3_v$ is stable, note that each agent gets matched to its second choice in $M^3_v$. Thus, the only potential blocking pair can be between an agent and its top choice. But no two agents consider each other as their top choice. So there is no blocking pair with respect to any of the three matchings.

    Any other perfect matching between $A_v\cup B_v$ either matches two agents in $A_v$ and hence two agents in $B_v$ to each other, which is blocked by every $(a^i_v,b^j_v)$ pair, $i,j\in [3]$, or it is one of the following matchings:
    \begin{eqnarray*}
        M^4_v &= & \{(a^1_v,b^1_v),(a^2_v,b^3_v),(a^3_v,b^2_v)\}\\
        M^5_v & = & \{(a^1_v,b^3_v),(a^2_v, b^2_v), (a^3_v,b^1_v) \}\\
        M^6_v & = & \{(a^1_v,b^2_v),(a^2_v,b^1_v),(a^3_v,b^3_v)\}
    \end{eqnarray*}
    They are respectively blocked by $(a^3_v,b^1_v), (a^1_v,b^2_v),(a^2_v,b^3_v)$.
    
\end{proof}

Now, we characterize stable matchings in $H$ in terms of the stable matchings in each gadget. Let 
\begin{proposition}\label{prop:global-roommates}
    A matching $M$ is stable in $H$ if and only if $M=\bigcup_{v\in V} M^{i_v}_v$ where $\forall v\in V~i_v\in [3]$, and $M^2_u\cup M^2_v\nsubseteq M$ for any $(u,v)\in E$.
\end{proposition}
\begin{proof}
    Let $M$ be a matching in $H$ such that, for some $(u,v)\in E$, $M^2_u\cup M^2_v\subseteq M$. Then both $a^1_u$ and $a^1_v$ prefer each other over their partners in $M$ viz. $b^3_u,b^3_v$ respectively, and hence $(a^1_u,a^1_v)$ blocks $M$. 

    Suppose $M$ matches some agent in $A_v\cup B_v$ to an agent different from its partner in any of the matchings $M^1_v,M^2_v,M^3_v$. Then $M$ is unstable by Proposition~\ref{prop:roommates}.
    
    Now, consider a matching $M$ that satisfies the condition in the proposition. It can be seen that $M$ is stable in $H$, as follows. There is no blocking pair within any gadget as $M$ matches each gadget according to one of the three stable matchings on that gadget. So, if there is a blocking pair with respect to $M$, it must contain two agents from two different gadgets. No agent among $(A_v\cup B_v)\setminus \{a^1_v\}$ for any $v_V$ can form a blocking pair with any agent from a different gadget because they prefer their partners in $M^1_v,M^2_v,M^3_v$ more than any agent outside their gadget. Also, $(a^1_u,a^1_v)$ can block $M$ only if $(u,v)\in E$ and, moreover $M$ contains the matchings $M^2_u,M^2_v$, violating the given condition. Thus, $M$ is stable.
\end{proof}

\begin{lemma}\label{lem:roommates}
    There are three disjoint stable matchings in $H$ if and only if $G$ is $3$-colorable.
\end{lemma}
\begin{proof}
    Let $G$ have a $3$-coloring $c:\{1,2,3\}\rightarrow V$. We construct three disjoint stable matchings $M^{(1)},M^{(2)},M^{(3)}$ in $H$ as follows. Clearly, by Proposition~\ref{prop:global-roommates}, each of the three matchings $M^{(1)},M^{(2)},M^{(3)}$, if they exist, must match each gadget as per distinct matchings from $M^1_v,M^2_v,M^3_v$ for each $v\in V$. Therefore, for each $v\in V$, if $c(v)=i$, $i\in [3]$, then $M^2_v\subseteq M^{(i)}$. $M^1(v)$ and $M^3_v$ are made subsets of the remaining two matchings arbitrarily. Since no pair $u,v$ where $(u,v)\in E$ has the same color, $M^2_u$ and $M^2_v$ are contained in two different matchings among $M^{(1)},M^{(2)},M^{(3)}$. Thus, the three matchings are stable by Proposition~\ref{prop:global-roommates}. They are disjoint as $M^1_v,M^2_v,M^3_v$ are disjoint.

    Now consider three disjoint stable matchings in $H$, say $M^{(1)},M^{(2)},M^{(3)}$. Then a $3$-coloring $c$ for $G$ is constructed as follows. For each $v\in V$, if $M^2_v\subseteq M^{(i)}$ then assign $c(v)=i$. Again, by Proposition~\ref{prop:global-roommates}, no edge is monochromatic. Hence $c$ is a proper $3$-coloring of $G$.
\end{proof}
This completes the proof of Theorem~\ref{thm:roommates}. Note that, if $G$ is not $3$-colorable, then the corresponding instance $H$ admits only two disjoint stable matchings.

\begin{corollary}\label{cor:roommates}
    Determining if a stable roommates instance admits three disjoint stable matchings or at most two disjoint stable matchings is NP-complete even when preference lists have length at most $7$.
\end{corollary}
\begin{proof}
    The \threecol\ problem is NP-complete on maximum degree $4$ graphs \cite{GareyJS1976}. So the reduction given above produces preference lists of length at most $7$ on a graph with maximum degree $4$.
\end{proof}
\section{Conclusion}\label{sec:conclusion}
We addressed the question of finding a
largest pairwise disjoint collection of solutions to the \SM\ problem. We showed
that such a collection can in fact be found in time \emph{linear} in the input
size. The collection of stable matchings that our algorithm finds has the
additional property that they form a \emph{chain} in the distributive lattice of
stable matchings. To the best of our knowledge this is the first work on finding
pairwise disjoint \emph{stable} matchings, though this question has received
much attention for ``plain'' matchings without the stability requirement.

\textcolor{black}{Subsequent developments on this problem have answered some of the questions
on the diverse collections of stable matchings in the \SM\ problem, posed in the conference
version of this paper. 
Addressing the disjoint matchings problem with notions of optimality other than stability 
is an interesting direction to pursue.
}

\section*{Acknowledgements}
PN and GP are partly supported by a grant from Infosys Foundation. PN is partly supported by ANRF/ARGM/2025/002769/MTR. HV is supported by JST ERATO Grant Number JPMJER2301, Japan.
\section*{AI Disclosure} 
The reduction in Theorem~\ref{thm:roommates} was obtained in an interaction with ChatGPT. The proof in the paper is completely written by the authors.

\bibliography{references} 

@Article{doi:10.1080/00029890.1962.11989827,
  author    = {D. Gale and L. S. Shapley},
  journal   = {The American Mathematical Monthly},
  title     = {College Admissions and the Stability of Marriage},
  year      = {1962},
  number    = {1},
  pages     = {9-15},
  volume    = {69},
  doi       = {10.1080/00029890.1962.11989827},
  publisher = {Taylor & Francis},
  url       = {https://doi.org/10.1080/00029890.1962.11989827},
}

@Article{teo1998geometry,
  author  = {Teo, Chung-Piaw and Sethuraman, Jay},
  journal = {Mathematics of Operations Research},
  title   = {The geometry of fractional stable matchings and its applications},
  year    = {1998},
  issn    = {0364765X},
  number  = {4},
  pages   = {874--891},
  volume  = {23},
}

@Book{Gusfield1989,
  author    = {Dan Gusfield and Robert W. Irving},
  publisher = {{MIT} Press},
  title     = {The Stable marriage problem - structure and algorithms},
  year      = {1989},
  isbn      = {978-0-262-07118-5},
  series    = {Foundations of computing series},
  bibsource = {dblp computer science bibliography, https://dblp.org},
}

@article{konig1916graphen,
  title={{\"U}ber graphen und ihre anwendung auf determinantentheorie und mengenlehre},
  author={K{\"o}nig, D{\'e}nes},
  journal={Mathematische Annalen},
  volume={77},
  number={4},
  pages={453--465},
  year={1916},
  publisher={Springer}
}

@article{tutte1954short,
  title={A short proof of the factor theorem for finite graphs},
  author={Tutte, William Thomas},
  journal={Canadian Journal of Mathematics},
  volume={6},
  pages={347--352},
  year={1954},
  publisher={Cambridge University Press}
}

@article{holyer1981np,
  title={The {NP}-completeness of edge-coloring},
  author={Holyer, Ian},
  journal={SIAM Journal on computing},
  volume={10},
  number={4},
  pages={718--720},
  year={1981},
  publisher={SIAM}
}

@article{DBLP:journals/corr/abs-2009-04567,
  author    = {Fedor V. Fomin and
               Petr A. Golovach and
               Lars Jaffke and
               Geevarghese Philip and
               Danil Sagunov},
  title     = {Diverse Pairs of Matchings},
  journal   = {CoRR},
  volume    = {abs/2009.04567},
  year      = {2020},
  url       = {https://arxiv.org/abs/2009.04567},
}

@article{mkrtchyan2008edge,
  title={On edge-disjoint pairs of matchings},
  author={Mkrtchyan, Vahan V and Musoyan, Vahe L and Tserunyan, Anush V},
  journal={Discrete mathematics},
  volume={308},
  number={23},
  pages={5823--5828},
  year={2008},
  publisher={Elsevier}
}

@article{faudree1991neighborhood,
  title={Neighborhood conditions and edge-disjoint perfect matchings},
  author={Faudree, Ralph J and Gould, Ronald J and Lesniak, Linda M},
  journal={Discrete mathematics},
  volume={91},
  number={1},
  pages={33--43},
  year={1991},
  publisher={Elsevier}
}

@article{10.2307/90020380,
 author = {Hongliang Lu and David G.L. Wang},
 journal = {Applicable Analysis and Discrete Mathematics},
 number = {1},
 pages = {11--38},
 publisher = {University of Belgrade, Serbia},
 title = {THE NUMBER OF DISJOINT PERFECT MATCHINGS IN SEMI-REGULAR GRAPHS},
 volume = {11},
 year = {2017}
}

@Book{Knuth76,
 author = {Knuth, Donald},
 title = {Mariages stables et leurs relations avec d'autres probl\`{e}mes combinatoires : introduction \`{a} l'analyse math\`{e}matique des algorithmes},
 publisher = {Presses de l'Universit\`{e} de Montr\`{e}al},
 year = {1976},
 address = {Montr\`{e}al},
 isbn = {0840503423}
 }

@article{McVitie71,
author = {McVitie, D. G. and Wilson, L. B.},
title = {The Stable Marriage Problem},
year = {1971},
publisher = {Association for Computing Machinery},
address = {New York, NY, USA},
volume = {14},
number = {7},
issn = {0001-0782},
journal = {Commun. ACM},
pages = {486?490},
}

@article{Gusfield87,
  title={Three Fast Algorithms for Four Problems in Stable Marriage},
  author={D. Gusfield},
  journal={SIAM J. Comput.},
  year={1987},
  volume={16},
  pages={111-128}
}

@article{Pittel89,
  author    = {Boris G. Pittel},
  title     = {The Average Number of Stable Matchings},
  journal   = {{SIAM} J. Discret. Math.},
  volume    = {2},
  number    = {4},
  pages     = {530--549},
  year      = {1989},
}

@article{LP09,
author = {Lennon, Craig and Pittel, Boris},
title = {On the Likely Number of Solutions for the Stable Marriage Problem},
year = {2009},
publisher = {Cambridge University Press},
address = {USA},
volume = {18},
number = {3},
issn = {0963-5483},
journal = {Comb. Probab. Comput.},
pages = {371?421},
}

@misc{palmer2020447n,
      title={At most $4.47^n$ stable matchings}, 
      author={Cory Palmer and Dömötör Pálvölgyi},
      year={2020},
      eprint={2011.00915},
      archivePrefix={arXiv},
      primaryClass={math.CO}
}

@article{doi:10.1137/0215048,
author = {Irving, Robert W. and Leather, Paul},
title = {The Complexity of Counting Stable Marriages},
journal = {SIAM Journal on Computing},
volume = {15},
number = {3},
pages = {655-667},
year = {1986},
doi = {10.1137/0215048},
}

@inproceedings{10.1145/3188745.3188848,
author = {Karlin, Anna R. and Gharan, Shayan Oveis and Weber, Robbie},
title = {A Simply Exponential Upper Bound on the Maximum Number of
Stable Matchings},
year = {2018},
publisher = {Association for Computing Machinery},
address = {New York, NY, USA},
doi = {10.1145/3188745.3188848},
booktitle = {Proceedings of the 50th Annual ACM SIGACT Symposium on
Theory of Computing},
pages = {920–925},
numpages = {6},
series = {STOC 2018},
}

@article{Irving1985,
  author    = {Robert W. Irving},
  title     = {An efficient algorithm for the ``stable roommates'' problem},
  journal   = {Journal of Algorithms},
  volume    = {6},
  number    = {4},
  pages     = {577--595},
  year      = {1985},
  doi       = {10.1016/0196-8858(85)90033-1}
}

@article{Gusfield1988,
  author    = {Dan Gusfield},
  title     = {The Structure of Stable Roommates Matchings, and a Polynomial Time Algorithm for Finding All Stable Peer Matchings},
  journal   = {SIAM Journal on Computing},
  volume    = {17},
  number    = {4},
  pages     = {742--769},
  year      = {1988},
  doi       = {10.1137/0217047}
}

@article{GareyJS1976,
  author    = {M. R. Garey and David S. Johnson and Larry J. Stockmeyer},
  title     = {Some simplified {NP}-complete graph problems},
  journal   = {Theoretical Computer Science},
  volume    = {1},
  number    = {3},
  pages     = {237--267},
  year      = {1976},
  doi       = {10.1016/0304-3975(76)90059-1}
}

@inproceedings{PalmerP2021,
  author    = {Cory Palmer and D{\"o}m{\"o}t{\"o}r P{\'a}lv{\"o}lgyi},
  title     = {At most $3.55^n$ stable matchings},
  booktitle = {Proceedings of the 2021 IEEE 62nd Annual Symposium on Foundations of Computer Science (FOCS)},
  pages     = {217--227},
  year      = {2021},
  publisher = {IEEE},
  doi       = {10.1109/FOCS52979.2021.00029}
}

@inproceedings{deBergOM2025,
  author    = {Mark de Berg and Tim A. E. Oosterwijk and Jannik Matuschke},
  title     = {Finding Diverse Solutions in Combinatorial Problems with a Distributive Lattice Structure},
  booktitle = {Proceedings of the 36th International Symposium on Algorithms and Computation (ISAAC 2025)},
  series    = {LIPIcs},
  volume    = {359},
  pages     = {11:1--11:14},
  year      = {2025},
  publisher = {Schloss Dagstuhl -- Leibniz-Zentrum f{\"u}r Informatik},
  doi       = {10.4230/LIPIcs.ISAAC.2025.11}
}

@misc{FleinerFK2024Approach,
  author        = {Tam{\'a}s Fleiner and Andr{\'a}s Frank and Tam{\'a}s Kir{\'a}ly},
  title         = {A new approach to bipartite stable matching optimization},
  year          = {2024},
  eprint        = {2409.04885},
  archivePrefix = {arXiv},
  primaryClass  = {cs.DS},
  url           = {https://arxiv.org}
}
\newpage
\appendix
\section{Extended Gale-Shapley Algorithm}\label{app-sec:exGS}
We give a formal pseudocode for the extended Gale-Shapley algorithm here.
\begin{algorithm}[]
  \caption{Extended Gale-Shapley}
  \label{algorithm:extended-gale-shapley}
  \begin{algorithmic}[1]
    \Procedure {GS-Extended}{$G$} \Comment {$G$ is an SM instance}
    \State {\texttt{assign each person to be free}}
    \While {\texttt{some man $m$ is free}}
    \State $w \gets $ first woman on $m$'s list
    \If {\texttt{some man $p$ is engaged to $w$}}
    \State \texttt{assign $p$ to be free}
    \EndIf
    \State \texttt{assign $m$ and $w$ to be engaged to each other}
    \For {\texttt{each successor $m'$ of $m$ on $w$'s list}}
    \State \texttt{delete $w$ on $m'$'s list}
    \State \texttt{delete $m'$ on $w$'s list} \Comment {deleting the pair $(m',w)$}
    \EndFor
    \EndWhile
    
    \Return \texttt{Stable matching consisting of $n$ engaged pairs}
    \EndProcedure
  \end{algorithmic}
\end{algorithm}

\section{Examples}\label{app-sec:ex}
\paragraph{Example instance with multiple disjoint stable matchings: }
\begin{table}[H]
  \centering
\begin{tabular}{cccccccccccccccc}
  \hline
 & 1  &  & 1 & 2 & 3 &  &  &  & 1 &  & 2 & 3 & 1  &  &\\
 & 2  &  & 2 & 3 & 1 &  &  &  & 2 &  & 3 & 1 & 2  &  &\\
 & 3  &  & 3 & 1 & 2 &  &  &  & 3 &  & 1 & 2 & 3  &  &\\
  \hline
 \multicolumn{8}{c}{Men's Preferences}   &   \multicolumn{8}{c}{Women's Preferences} \\
  \hline 
\end{tabular}
\captionof{figure}{A stable marriage instance of size 3.}
\end{table}
\label{fig:marriage_instance_3.1}

It can be easily verified that the above marriage instance has three (and only three) disjoint matchings as given below.

\begin{center}
  \label{fig:disjoint_stable}
  \begin{tikzpicture}
    \draw[fill=black] (0,0) circle (2pt);
    \draw[fill=black] (0,1) circle (2pt);
    \draw[fill=black] (0,2) circle (2pt);
    \draw[fill=black] (2,0) circle (2pt);
    \draw[fill=black] (2,1) circle (2pt);
    \draw[fill=black] (2,2) circle (2pt);
    \node at (-0.5,0) {$m_3$};
    \node at (-0.5,1) {$m_2$};
    \node at (-0.5,2) {$m_1$};
    \node at (2.5,0) {$w_3$};
    \node at (2.5,1) {$w_2$};
    \node at (2.5,2) {$w_1$};
    
    \draw[thick] (0,0) -- (2,0);
    \draw[thick] (0,1) -- (2,1);
    \draw[thick] (0,2) -- (2,2);
  \end{tikzpicture}
  \hspace{2pt}
  \begin{tikzpicture}
    \draw[fill=black] (0,0) circle (2pt);
    \draw[fill=black] (0,1) circle (2pt);
    \draw[fill=black] (0,2) circle (2pt);
    \draw[fill=black] (2,0) circle (2pt);
    \draw[fill=black] (2,1) circle (2pt);
    \draw[fill=black] (2,2) circle (2pt);
    \node at (-0.5,0) {$m_3$};
    \node at (-0.5,1) {$m_2$};
    \node at (-0.5,2) {$m_1$};
    \node at (2.5,0) {$w_3$};
    \node at (2.5,1) {$w_2$};
    \node at (2.5,2) {$w_1$};
    
    \draw[thick] (0,0) -- (2,2);
    \draw[thick] (0,1) -- (2,0);
    \draw[thick] (0,2) -- (2,1);
  \end{tikzpicture}
  \hspace{2pt}
  \begin{tikzpicture}
    \draw[fill=black] (0,0) circle (2pt);
    \draw[fill=black] (0,1) circle (2pt);
    \draw[fill=black] (0,2) circle (2pt);
    \draw[fill=black] (2,0) circle (2pt);
    \draw[fill=black] (2,1) circle (2pt);
    \draw[fill=black] (2,2) circle (2pt);
    \node at (-0.5,0) {$m_3$};
    \node at (-0.5,1) {$m_2$};
    \node at (-0.5,2) {$m_1$};
    \node at (2.5,0) {$w_3$};
    \node at (2.5,1) {$w_2$};
    \node at (2.5,2) {$w_1$};
    
    \draw[thick] (0,0) -- (2,1);
    \draw[thick] (0,1) -- (2,2);
    \draw[thick] (0,2) -- (2,0);
  \end{tikzpicture}
  \captionof{figure}{Disjoint stable matchings $M_0, M_1$ and $M_z$}
\end{center}

\paragraph{Example instance with multiple chains}

\begin{table}[H]
  \centering
\begin{tabular}{|p{5mm}|p{5mm}p{5mm}p{5mm}p{5mm}p{5mm}p{5mm}|p{5mm}p{5mm}p{5mm}p{5mm}|}
  \hline
  1 & 4 & 1 & 3 & 2 & & 1 & 2 & 1 & 3 & 4\\
  2 & 4 & 2 & 3 & 1 & & 2 & 1 & 3 & 2 & 4\\
  3 & 1 & 3 & 2 & 4 & & 3 & 4 & 2 & 3 & 1\\
  4 & 1 & 4 & 2 & 3 & & 4 & 3 & 4 & 2 & 1\\
  \hline
 \multicolumn{5}{c}{Men's Preferences}   &   \multicolumn{6}{c}{Women's Preferences} \\
  \hline 
\end{tabular}
\captionof{figure}{A stable marriage instance of size 4.}
\end{table}
\label{fig:marriage_instance_4.1}

The above marriage instance has only three stable matchings. They are as follows.

\begin{center}
  \label{fig:disjoint_stable-list}
  \begin{tikzpicture}
    \draw[fill=black] (0,0) circle (2pt);
    \draw[fill=black] (0,1) circle (2pt);
    \draw[fill=black] (0,2) circle (2pt);
    \draw[fill=black] (0,3) circle (2pt);
    \draw[fill=black] (2,0) circle (2pt);
    \draw[fill=black] (2,1) circle (2pt);
    \draw[fill=black] (2,2) circle (2pt);
    \draw[fill=black] (2,3) circle (2pt);
    \node at (-0.5,0) {$m_4$};
    \node at (-0.5,1) {$m_3$};
    \node at (-0.5,2) {$m_2$};
    \node at (-0.5,3) {$w_1$};
    \node at (2.5,0) {$w_4$};
    \node at (2.5,1) {$w_3$};
    \node at (2.5,2) {$w_2$};
    \node at (2.5,3) {$w_1$};
    
    \draw[thick] (0,0) -- (2,0);
    \draw[thick] (0,1) -- (2,1);
    \draw[thick] (0,2) -- (2,2);
    \draw[thick] (0,3) -- (2,3);
  \end{tikzpicture}
  \hspace{2pt}
  \begin{tikzpicture}
    \draw[fill=black] (0,0) circle (2pt);
    \draw[fill=black] (0,1) circle (2pt);
    \draw[fill=black] (0,2) circle (2pt);
    \draw[fill=black] (0,3) circle (2pt);
    \draw[fill=black] (2,0) circle (2pt);
    \draw[fill=black] (2,1) circle (2pt);
    \draw[fill=black] (2,2) circle (2pt);
    \draw[fill=black] (2,3) circle (2pt);
    \node at (-0.5,0) {$m_4$};
    \node at (-0.5,1) {$m_3$};
    \node at (-0.5,2) {$m_2$};
    \node at (-0.5,3) {$w_1$};
    \node at (2.5,0) {$w_4$};
    \node at (2.5,1) {$w_3$};
    \node at (2.5,2) {$w_2$};
    \node at (2.5,3) {$w_1$};
    
    \draw[thick] (0,0) -- (2,0);
    \draw[thick] (0,1) -- (2,2);
    \draw[thick] (0,2) -- (2,1);
    \draw[thick] (0,3) -- (2,3);
  \end{tikzpicture}
  \begin{tikzpicture}
    \draw[fill=black] (0,0) circle (2pt);
    \draw[fill=black] (0,1) circle (2pt);
    \draw[fill=black] (0,2) circle (2pt);
    \draw[fill=black] (0,3) circle (2pt);
    \draw[fill=black] (2,0) circle (2pt);
    \draw[fill=black] (2,1) circle (2pt);
    \draw[fill=black] (2,2) circle (2pt);
    \draw[fill=black] (2,3) circle (2pt);
    \node at (-0.5,0) {$m_4$};
    \node at (-0.5,1) {$m_3$};
    \node at (-0.5,2) {$m_2$};
    \node at (-0.5,3) {$w_1$};
    \node at (2.5,0) {$w_4$};
    \node at (2.5,1) {$w_3$};
    \node at (2.5,2) {$w_2$};
    \node at (2.5,3) {$w_1$};
    
    \draw[thick] (0,0) -- (2,1);
    \draw[thick] (0,1) -- (2,0);
    \draw[thick] (0,2) -- (2,3);
    \draw[thick] (0,3) -- (2,2);
  \end{tikzpicture}
  \captionof{figure}{Stable matchings $M_0, M_1$ and $M_3$}
\end{center}


\end{document}